\documentclass[leqno,12pt]{article} 
\usepackage{amsmath, amssymb}
\usepackage{amsthm} 
\usepackage{amscd}
\usepackage{amsxtra}
\usepackage{xy}
\usepackage{xypic}
\usepackage{tikz,varwidth}
\usepackage{xparse}
\usepackage{ytableau}
\usepackage{xcolor}

\newcount\tableauRow
\newcount\tableauCol

\newenvironment{Tableau}[1]{%
  \tikzpicture[scale=0.4,draw/.append style={thick,black},
                      baseline=(current bounding box.center)]
    \tableauRow=-1.5
    \foreach \Row in {#1} {
       \tableauCol=0.5
       \foreach\k in \Row {
         \draw[thin](\the\tableauCol,\the\tableauRow)rectangle++(1,1);
         \draw[thin](\the\tableauCol,\the\tableauRow);
         \global\advance\tableauCol by 1
       }
       \global\advance\tableauRow by -1
    }
}{\endtikzpicture}

\define\cal{\mathcal}

\define\Z{\mathbb Z}
\define\N{\mathbb N}

\theoremstyle{plain} 
\newtheorem{theorem}{\indent\sc Theorem}[section]
\newtheorem{lemma}[theorem]{\indent\sc Lemma}
\newtheorem{corollary}[theorem]{\indent\sc Corollary}

\theoremstyle{definition} 

\newtheorem{remark}[theorem]{\indent\sc Remark}
\newtheorem{example}[theorem]{\indent\sc Example}

\newtheorem{problem}[theorem]{\indent\sc Problem}

\makeatletter
\def\address#1#2{\begingroup
\noindent\parbox[t]{7.8cm}{%
\small{\scshape\ignorespaces#1}\par\vskip1ex
\noindent\small{\itshape E-mail address}%
\/: #2\par\vskip4ex}\hfill%
\endgroup}%
\makeatother
\title{\uppercase{Arithmetic of Heisenberg ring and cyclic group actions }} 
\author{
\bigskip \\
\textsc{ Piotr Kraso{\'n}, Jan Milewski } 
}
\date{} 
\begin{document}
\maketitle



\begin{abstract}
\noindent
In this paper we compute in some new cases the  cardinalities of the fibers of certain natural fibrations that appear in the analysis of the configuration space of the Heisenberg ring.  This is done by means of certain cyclic group actions on some subsets of 
restricted partitions. 
   \end{abstract}

\section{Introduction}
In this paper we consider the Heisenberg ring with $N$  nodes and   $r\leq \frac{N}{2}$ reversed spins. According to this model at each node at most one reverse is admissible.
 Any such reverse is called a Bethe pseudo-particle. By assumption Bethe  pseudo-partcles are identical and indistinguishable.
It is natural to treat the set of nodes as a  cyclic group ${\mathbb Z}/n{\mathbb Z}$ or rather a regular orbit of it.
Thus
\begin{equation} \label{tildan}
\tilde{N}=\{j|j=1,\ldots ,N\}\equiv \Z/N\Z, \quad {\mathrm{and}}, \quad [j]_N\in \Z/N\Z, \quad      1\leq j\leq N.\end{equation}

Classical configuration space of the Heisenberg ring with $r$ Bethe pseudo-particles can be identified with the following set:

\begin{equation} Q_{N,r} =\{ {\bf j}= (j_1, \ldots , j_r)\in {\tilde N}^r |\;  1\leq j_1< \ldots < j_r \leq N \} \end{equation}
which corresponds, under the identification (\ref{tildan}), to the family  of $r$ element subsets of 
 $\Z/n\Z$. 
 
 On the covering space:
 \begin{equation}  P_{N,r}= \{ (j_1, \ldots , j_r)\in {\Z}^r |\; j_1< \ldots < j_r <j_1+N \}   \end{equation}
we introduce 
an action $ L \, : \, ({\Z},+) \times P_{N,r} \longrightarrow P_{N,r}$ defined by the following condition:
\begin{equation} L(1,(j_1, \ldots , j_r))=(j_2, \ldots , j_r,j_1+N) .\end{equation} 
The orbit space of  this action is equal to the classical configuration space:
\begin{equation}	Q_{N,r}\equiv P_{N,r}/L .\end{equation}
The following projections:

\begin{equation} \pi _{N,r}: Q_{N,r}\longrightarrow {\mathbb Z}/N{\mathbb Z}, \quad
\label{smb} \pi _{N,r}({\bf j})=\left[ \sum_{\alpha \in \tilde r} j_{\alpha} \right]_N, \end{equation}

\begin{equation} \tilde \pi _{N,r}: P_{N,r}\longrightarrow {\Z}, \quad
\label{smt} \tilde \pi _{N,r}({\bf j})= \sum_{\alpha \in \tilde r} j_{\alpha} \end{equation} 
can be interpreted as suitably rescaled centers of masses of configurations. More precisely, they are multiplied by the number of Bethe pseudo-particles.
The set of all relative positions 
\begin{equation} \Delta _{N,r}=\{{\bf t} \in {\N} ^r|\;  \sum_{\alpha \in \tilde r} t_{\alpha} \, = \, N   \} \end{equation} 
where  ${\bf t} ={\bf t}({\bf j})=(t_1, \ldots ,t_r)$ is the vector of relative positions with the following components:   
\begin{equation} \label{wspwzgl}  t_{\alpha }  = \begin{cases} 
j_{\alpha +1}-j_{\alpha} \;\;\;\;\;\;\; {\mathrm{for}} \;\; \; 1 \leq \alpha <r, \\
N+j_1-j_r  \,\;\;\; {\mathrm{for}} \; \;\; \alpha = r
\end{cases},
\end{equation} 
forms a lattice simplex:
\begin{equation} \Delta _{N,r}=\{{\bf t} \in {\N} ^r|\;  \sum_{\alpha \in \tilde r} t_{\alpha} \, = \, N   \}. \end{equation} 
Each configuration
 ${\bf j}\in P_{N,r}$ can be uniquely determined  by means of the aforementioned rescaled center of mass  $s=\pi _{N,r}({\bf j})$ and the relative configuration ${\bf t}({\bf j})$:
\begin{equation}  \label{odwr} j_\alpha (s,{\bf t}) = \frac{1}{r}(s+\sum_{\beta \in {\tilde r}}  \beta t_\beta) -\sum_{\alpha \leq \beta \leq r} t_\beta \; ,
\quad \alpha \in \tilde r .\end{equation}
From this we obtain the compatibility condition for 
 $s$ and ${\bf t:}$
\begin{equation} \label{zgodnosc} s+\sum_{\beta \in {\tilde r}}  \beta t_\beta \equiv 0 \mod r . \end{equation}

 The  projections $\pi _{N,r}$ resp. $\tilde{\pi} _{N,r}$  are fibrations with the total spaces spaces $Q_{N,r}$ 	resp. $P_{N,r}$ . 
We have 
\begin{equation} \label{decob} Q_{N,r} =\bigcup _{s\in {\bf Z}/N{\mathbf Z}} Q_{N,r,s}, \quad Q_{N,r,s}:=\{{\bf j}\in  Q_{N,r}| \pi _{N,r}({\bf j})=s \} \end{equation}
and
\begin{equation} \label{decot} P_{N,r} =\bigcup _{s\in {\bf Z}} P_{N,r,s},\quad P_{N,r,s}:=\{ {\bf j}\in  P_{N,r}| {\tilde \pi} _{N,r}({\bf j})=s \}.\end{equation}

In \cite{phstsol} it was shown that for any  $N\geq 2$ the configuration space $Q_{N,r}$ can be naturally embedded,  in a natural way, in a model manifold       $Q_r(S)=\{ A \subset S \,| \,\,card A =r \}.$ 
An analogous to considered above  coverings and fibrations were also given. These constructions resemble the fibrations considered by H. R. Morton \cite{Mor} and D.Tymoczko \cite{Tym1}  (cf. also  \cite{Tym})
who applied them in the theory of music chords.  They analyze configurations with repetitions which lead to orbifolds. Since we consider configurations without repetitions, our approach yields manifolds.


With these fibrations one can associate one more fibration (cf. the compatibility condition (\ref{zgodnosc})):
\begin{equation} \label{decob} \Delta_{N,r} =\bigcup _{s\in {\bf Z}/N{\mathbf Z}} \Delta_{N,r,s}, \quad \Delta_{N,r,s}:=\{ {\bf t} \in  \Delta_{N,r}| 
s+\sum_{\beta \in {\tilde r}}  \beta t_\beta \equiv 0 \mod r \}. \end{equation}
In the paper \cite{Mi} the number of elements in these fibers for some values of $N$  were computed. It was shown, using certain  natural group action, that for 
$N,r$  coprime the fibers have the same number of elements.
Similarly, for $s\neq 0 \mod r$ all fibers $\Delta_{Mp,p,s}$  are of equal cardinality  for  a rational prime 
  $p>2$ and  $M>1$. For $p=3,5$  and $s\neq 0 \mod r$, it was additionally shown that
\begin{equation} \Delta_{Mp,p,0}=\Delta_{Mp,p,s}+1.\end{equation}
In this paper we show that the above formula holds for any prime $p>2.$

We phrase the relations  in terms of the relation among coefficients of Gaussian polynomials (cf. Remark \ref{rmg} and Lemma \ref{bijection}).
The cardinalities of the fibers $|\Delta_{Mp,p,s}| , \, s=0,\dots, p-1$  follow  then  directly from Theorem \ref{therm}

 For known relations for coefficients of Gaussian polynomials and history see \cite{Ex}. In this short note we give a direct combinatorial proof of some relations between the sums of coefficients of Gaussian polynomials. We prove that if $\gcd(k,l)=1$ then the sums over  moduli classes modulo $l$ of coefficients of   Gaussian polynomials corresponding to the rectangle $k\times(l-1)$ are all equal (cf. Theorem \ref{main1}).  We also prove one case where $ \gcd (k,l)\neq 1.$ Namely for an odd prime number $p$ and $k=Mp, l=p$ we compute the corresponding sums. This result is not true when $p$ is not a prime.  
The main tool is the action of the additive group ${\mathbb Z}/p{\mathbb Z}$ on certain subsets of the restricted partitions. 
In section 2 we very briefly give necessary definitions. In section 3 we define some group actions on the sets of restricted partitions.   We also give some comments and state a general problem.

\section{Preliminaries}
In this section we set up notation and recall basic facts concerning partitions. The reader is advised to consult one of the excellent references on the subject \cite{An}, \cite{An1},\cite{Stan1}, \cite{Stan2}.
Let ${\frak p}(N,M; n)$ denote the number of partitions of a natural number $n$ into at most $M$ parts, each of size not exceeding $N$. Geometrically that means that the Young diagram corresponding to such a partition $\pi =<{\pi}_1,\dots,
{\pi}_l> \,, \quad  \, {\pi}_{1}\ge , \dots ,\ge {\pi}_{l}$ fits inside an $M\times N$ rectangle.
 The following recurrence formula holds
\begin{equation}\label{21}
{\frak p}(M,N,n)={\frak p}(M,N-1,n)+{\frak p}(M-1,N,n-N)
\end{equation}

Formula  (\ref{21}) comes from the observation that there is a bijection between the  partitions of $n$ into exactly $N$ parts of size not exceeding $M$ and the partitions of $n-N$ into at most N parts of size not exceeding $M-1.$
The bijection is given by subtracting $1$ from each part of the former.

The Gaussian coefficient may be considered as a generalization of the usual binomial coefficient and is defined by the following formula:
\begin{equation}\label{22}
G(L,K; q)={{K+L}\brack{K}}_{q}=\frac{{\prod}_{j=1}^{K+L}(1-q^j)}{{\prod}_{j=1}^{K}(1-q^j){\prod}_{j=1}^{L}(1-q^j)}
\end{equation}
The Gaussian binomial coefficient yields a generating function of restricted partitions
\begin{equation}\label{23}
{\sum}_{n=0}^{MN}{\frak p}(M,N,n)q^{n} = {{M+N}\brack{N}}_{q}
\end{equation}
In particular (\ref{23}) shows that  the Gaussian binomial coefficient is a polynomial with integer coefficients.
Equation (\ref{21})  yields the following
\begin{equation}\label{24}
{{N}\brack{M}}_{q}= {{N-1}\brack{M-1}}_{q} + q^{M}{{N-1}\brack{M}}_{q}
\end{equation}

\section{Some group actions}

Let $k,l,r$ be natural numbers and $i\in \{0,1,\dots ,r-1 \}.$
Define the following sets
\begin{equation}\label{a1}
\begin{split}
^{r}T_{k,l}^{i}:= & \, \{ \,{\text{\rm the set of restricted partitions of a number congruent to}}\quad  i\mod r  \\
 & {\text{\rm that fit into the rectangle  \,\, $k\times l$  \,  }}\}
 \end{split}
\end{equation}

Let 
\begin{equation}\label{aaa}
T_{k,l}={^r}T_{k,l}^{0}\sqcup {^r}T_{k,l}^{1}\sqcup \dots \sqcup {^r}T_{k,l}^{r-1}
\end{equation}

be the set of partitions that fit into the rectangle $k\times l.$ Notice that $T_{k,l}$ is independent of $r$ and of course $|T_{k,l}|=|T_{l,k}|.$
Let further 
\begin{equation}\label{4}
{^r}{\Omega}^{i}_{k,l}=\{ {\text{\small\rm non-increasing  surjections}} \,\, f\hspace{-2pt}: [0,k] \rightarrow \{ 1,\dots , l \} :  {\int}_{0}^{k}f dx \equiv \hspace{-4pt}\,i \hspace{-8pt}\mod \, r        \}.
\end{equation}
and ${\Omega}_{k,l}={^r}{\Omega}^{0}_{k,l}\sqcup \dots \sqcup {^r}{\Omega}^{r-1}_{k,l}.$  In ({\ref{4}}) $[0,k] \subset {\mathbb R}$ denotes the closed interval. 
\begin{remark}\label{rmg0}
Notice that  
${\Omega}_{k,l}$ is the set of all non-increasing surjections $f: [0,k] \rightarrow  \{ 1,\dots , l \}$ and thus is independent of $r.$
\end{remark}
\begin{remark}\label{rmg}
Notice that by (\ref{decob})   we have  a bijection between ${\Delta}_{N,r,s}$ and ${^r}{\Omega}^{r-s}_{N,r}$  and therefore $|{\Delta}_{N,r,s}| = |{^r}{\Omega}^{r-s}_{N,r}|.$
\end{remark}
We start with the following lemma.
\begin{lemma}\label{bijection}
Let ${m=l(l-1)/2}+k+l +i$ There is a bijection between the following sets
\begin{equation}\label{bb}
f: {^r}{\Omega}_{k+l,l}^{m} \rightarrow {^r}T_{k,l-1}^{i}
\end{equation}

\end{lemma}
\begin{proof}
Any non-increasing surjection in ${\Omega}_{k+l,l}$ is uniquely defined by a sequence of natural numbers $0< t_{l} < t_{l-1} < \dots < t_{2} < t_{1} = k+l .$
Order these sequences lexicographically.  Among them there is a minimal sequence (in the left diagram of  Fig.1  the minimal surjection corresponds to the grey region, for k= 11, l=4 ). There is also a maximal element which corresponds to the border of the crossed region.
Any non-increasing surjection has to have its graph inside the region between the grey and crossed regions. In Fig.1 we depicted the surjection corresponding to the sequence $(t_{4}, t_{3}, t_{2}, t_{1})=(4,5,7,15 ).$ We define the map $f$ in the following way
\begin{equation}\label{cor}
f( (t_{l}, \dots , t_{2}, t_{1}))=< t_{2}-(l-1), \dots , (t_{l}-1)> = <{\pi}_{1},\dots {\pi}_{l-1}>
\end{equation}
where the sequence in the brackets denotes the partition assigned to a non-increasing surjection (cf. Fig.1). 
This gives the bijection (\ref{bb}) and also a bijection between ${\Omega}_{k+l,l}$ and $T_{k,l-1}.$
\vspace{5mm}

\vspace{5mm}
\resizebox{\textwidth}{!}{ 
\begin{tikzpicture}

\node (n) 
{\begin{varwidth}{35cm}
{

\begin{ytableau}
*(lightgray) \bullet& \bullet& \bullet &\bullet & & & & & & & & & {\times}& {\times}& {\times} &\none &\bullet& \bullet&\bullet  &  && & & & & & \\
*(lightgray) \bullet & *(lightgray)  \bullet &\bullet &\bullet &\bullet   & & & & & & & & &{\times} &{\times} &\none &\bullet  &\bullet & \bullet &  & & & & & & &\\
*(lightgray) \bullet    & *(lightgray)  \bullet & *(lightgray) \bullet  &\bullet &\bullet   &\bullet &\bullet & & & & & & & &{\times}  &\none &\bullet& \bullet& \bullet &\bullet&& & & & &  & \\
 *(lightgray)   \bullet  & *(lightgray) \bullet  & *(lightgray) \bullet  &*(lightgray)  \bullet &*(lightgray)  \bullet   &*(lightgray)  \bullet  &*(lightgray)  \bullet  &*(lightgray) \bullet   & *(lightgray)  \bullet & *(lightgray) \bullet  & *(lightgray)  \bullet & *(lightgray)  \bullet &*(lightgray)  \bullet  & *(lightgray) \bullet  &*(lightgray)  \bullet &\none  &\none &\none&\none & \none & \none& \none& \none&\none & \none &\none \\
\end{ytableau}}
\end{varwidth}
};
  \end{tikzpicture}

  }
\vspace{5mm}

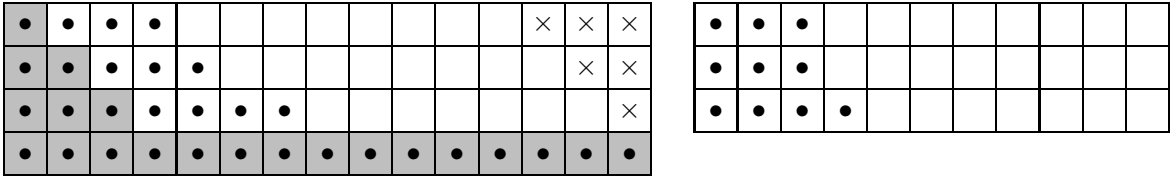
\centerline{ Fig.1 }

\end{proof}

Notice that the partition $\bf 0 \in\rm {^r}T_{k,l}^0$, for any $r$, and  corresponds to the minimal surjection.

\begin{corollary}\label{cr1}
The number of elements in the set $T_{k,l}$
\begin{equation}\label{no}
|T_{k,l}|= {{k+l}\choose {l}}
\end{equation}
\end{corollary}
\begin{proof}
By Lemma \ref{bijection} $|T_{k,l}|=|{\Omega}_{k+l+1,l+1}|$ and elements of ${\Omega}_{k+l+1,l+1}$ correspond bijectively to the sequences $0 < t_{l+1} < t_{l}<\dots < t_{2} < k+l+1$
\end{proof}

Let a sequence $0<t_l < t_{l-1} < \dots <t_2 <k+l$ correspond to an element $f\in {\Omega}_{k+l,l}$ as in the proof of Corollary \ref{cr1} .

Let
 \begin{equation}\label{seqn}
 (n_{1}, n_{2}, \dots , n_{l})\quad 
 \text{\rm where}\quad
 n_1=t_{l}, n_{2}=t_{l-1}-t_{l}, \,\dots ,\,n_{l}=k+l-t_2 
 \end{equation}
  be the corresponding sequence of lengths of   subintervals of $[0,k+l].$ Notice that assignment (\ref{seqn}) gives a bijection between the following sets
  of sequences
  \begin{equation}\label{tseq}
  (t_{2}, \dots ,t_{l})\quad 
  \text{\rm such that}\quad
  0<t_l < t_{l-1} < \dots <t_2 <k+l 
\end{equation}
and
 \begin{equation}\label{tseq}
  (n_{1}, \dots ,n_{l})\quad 
  \text{\rm such that}\quad
  n_{i}>0\quad
   \text{\rm and}\quad
  n_{1}+n_{2}+\dots n_{l}=k+l.
\end{equation}
We have the following  group actions on the sets ${\Omega}_{k+l,l}$. 
We define them on the sequences (\ref{seqn}).

{\bf Action of the cyclic group ${\mathbb Z}/{l{\mathbb Z}}$}

Let $C$ be a generator of the cyclic group ${\mathbb Z}/{l{\mathbb Z}}.$ We define the action of ${\mathbb Z}/{l{\mathbb Z}}$ on the sequence
$ (n_1, n_2, \dots , n_l)$ by the following  formula
\begin{equation}\label{act1}
C{\cdot} (n_1, n_2, \dots , n_l) = (n_{l}, n_1, \dots , n_{l-1})
\end{equation}

{\bf Action of the group of units $({\mathbb Z}/{l{\mathbb Z}})^{*}$}
 
Let $u\in ({\mathbb Z}/{l{\mathbb Z}})^{*}.$ Define the action of $u$ on the sequence $(n_1,\dots n_{l})$  by the following formula
\begin{equation}\label{act2}
u{\cdot} (n_1, n_2, \dots , n_l) = (n_{u^{-1}\cdot 1}, n_{u^{-1}\cdot 2}, \dots , n_{u^{-1}\cdot l})
\end{equation}

{\bf Action of the symmetric group ${\cal S}_{l}$}

Let ${\cal S}_{l}$ be the symmetric group on $l$ letters we define the action of ${\sigma}\in {\cal S}_{l}$ on a sequence $(n_{1}, n_{2},\dots , n_{l})$
by the formula 
\begin{equation}\label{act3}
{\sigma}\cdot (n_{1}, n_{2}, \dots , n_{l}) = (n_{{\sigma}(1)}, n_{{\sigma}(2)}, \dots n_{{\sigma}(l)})
\end{equation}

These actions, via bijection $g: {\Omega}_{k+l,l} \rightarrow T_{k,l-1}$ of Lemma \ref{bijection}, yield the corresponding action on the set of restricted partitions $T_{k,l-1}$. This transfer of actions is given by the following formula
\begin{equation}\label{transfer}
{\alpha}\cdot<{\pi} _{1}, \dots , {\pi}_{v} >  := g ({\alpha} \cdot g^{-1}( <{\pi} _{1}, \dots , {\pi}_{v} > )),
\end{equation}
where $\alpha$ is an element of an appropriate group.

These  actions on the sets of restricted partitions are of their own interest. We shall use the action (\ref{act1}) for obtaining some equalities concerning coefficients of Gaussian polynomials. We shall also use a different  action of the additive group ${\mathbb Z}/p{\mathbb Z}$ on certain subsets of restricted partitions to obtain similar equalities for odd prime numbers.

\section{Main results}
We have the following lemma 
\begin{lemma}\label{lm1}
Let  $k,l>0$ be natural numbers. If $\gcd(k,l)=1$ then for any $i,j \in {\mathbb Z}/l{\mathbb Z}$ we have $|{^l}{\Omega}_{k+l,l}^{i}|=|{^l}{\Omega}_{k+l,l}^{j}|.$
\end{lemma}
\begin{proof}
We use the action (\ref{act1}). Notice that for $f\in {^l}{\Omega}_{k+l,l}$ we have $ {\int}_{0}^{k+l}f dx = n_{1}\cdot l +n_{2}\cdot(l-1)+\dots +n_{l} .$ 
Thus if $n_{1}\cdot l +n_{2}\cdot(l-1)+\dots +n_{l} \equiv \, i \, \mod \, l$ then 
\begin{equation}\label{cng1}
n_{l}\cdot l +n_{1}\cdot (l-1) + \dots + n_{l-1} \equiv n_{1}\cdot l +n_{2}\cdot(l-1)+\dots +n_{l} -(n_{1}+\dots + n_{l}) 
\end{equation}
$$\equiv i-(k+l) \, \mod \, l .$$
Since $\gcd(k+l,l)=1,$ we see that the  set of congruence classes of integrals of the functions in the orbit $\{f, Cf, \dots C^{l-1}f\}$  correspond bijectively 
to ${\mathbb Z}/l{\mathbb Z}.$ This proves the lemma.
\end{proof}
Notice that Lemma \ref{lm1}, although formulated in different context, is similar to Lemma 1 of \cite{Mi}. 

\begin{corollary}\label{corro1}
Let  $k,l>0$ be natural numbers such that $\gcd(k,l)=1$ and let $r$ be a divisor of $l.$ Then $|{^r}{\Omega}_{k,l}^{i}|=|{^r}{\Omega}_{k,l}^{j}|$ for any 
$i,j \in {\mathbb Z}/r{\mathbb Z}.$
\end{corollary}
\begin{proof}
Apply Lemma \ref{lm1} and the reduction map ${\mathbb Z}/l{\mathbb Z}\rightarrow {\mathbb Z}/r{\mathbb Z}.$
\end{proof}

\begin{theorem}\label{main1}
Let $k,l$ be positive integers such that $\gcd(k,l)=1$ and $r$ any positive divisor of $l.$ For any $i\in {\mathbb Z}/r{\mathbb Z}$ we have
\begin{equation}\label{sumar}
{\sum}_{m=0}^{k(l-1)}{\frak p}(k, l-1, j\equiv  i  \mod\, r) = {\frac{1}{r}}{{k+l-1}\choose{l-1}} 
\end{equation}
\end{theorem}
\begin{proof}
Corollary \ref{corro1} and Lemma \ref{bijection} show that under the assumption of the theorem for any $i,j \in {\mathbb Z}/r{\mathbb Z}$ 
we have $|{^r}T_{k,l-1}^{i}|=|{^r}T_{k,l-1}^{j}|.$ The left hand side of (\ref{sumar}) is equal to $|{^r}T_{k,l-1}^{i}|.$ Now the theorem follows 
from Corollary \ref{cr1}.
\end{proof}
\begin{example}\label{ex0}
For $k=3,l=4$ and $r=4$ we have $T_{k,l-1}^{i}=5$ for $i=0,1,2,3.$
\end{example}
\begin{example}\label{ex10}
For $k=5,l=3$ and $r=3$ we have $T_{k,l-1}^{i}=7$ for $i=0,1,2.$
\end{example}

\begin{remark}\label{remk}
Notice that the left hand side of (\ref{sumar}) is the sum ${\sum}_{{n\equiv r} ({\text{mod}} p)} a_{n}$ over the  moduli classes modulo $r$ of the coefficients  of the Gaussian polynomial $G(k, l-1, q)={\sum}_{n=0}^{k(l-1)}a_{n}q^{n} .$
\end{remark}

\begin{remark}\label{remk1}
Theorem \ref{main1} is false if $\gcd(k,l)\neq 1$ as the following examples show.
\end{remark}
\begin{example}\label{ex1}
Let $k=l=6$ and $r=6.$ We have
$$ T_{6.5}^{0}=80, \,\,T_{6.5}^{1} = 75, \,\, T_{6.5}^{2}=78, \,\, T_{6.5}^{3}=76, \,\, T_{6.5}^{4}=78, \,\, T_{6.5}^{5}=75.$$
\end{example}
\begin{example}\label{ex2}
Let $k=20, l=10$ and $r=10.$ We have
$$ T_{10,.9}^{0}=9252, \,\,T_{10,9}^{1} = 9225, \,\, T_{10,9}^{2}=9250, \,\, T_{10,9}^{3}=9225, \,\, T_{10.9}^{4}=9250, $$
$$ T_{10.9}^{5}=9226,\,\, T_{10,.9}^{6}=9250, \,\,T_{10,9}^{7} = 9225, \,\, T_{10,9}^{8}=9250, \,\, T_{10,9}^{9}=9225.$$
\end{example}
\begin{remark}\label{nrem}
Symmetries in the numbers $T_{k,l}^j$ in Examples \ref{ex1} and \ref{ex2} give us some information about orbits of the action of the group $({\mathbb Z}/l{\mathbb Z})*$ (cf. (\ref{act2}) and (\ref{transfer})).  In Example \ref{ex1}  $({\mathbb Z}/6{\mathbb Z})^{*}\cong {\mathbb Z}/2{\mathbb Z}$ and we have 
$80 +76$ orbits of length $1$ and $75+78$ orbits of length $2.$ Similarly, in Example \ref{ex2}  $({\mathbb Z}/10{\mathbb Z})^{*}\cong {\mathbb Z}/4{\mathbb Z}$ and we see that we have $9252+9226$ orbits of length $1$ and $9925+9250$ orbits of length $4.$ The fact that $T_{k,l-1}^{j}=T_{k,l-1}^{l-j}$,
which can easily be justified by considering  complement Young diagrams or symmetry in the formula (\ref{22}).
\end{remark}

Examples \ref{ex1} and \ref{ex2} show that when $\gcd(k,l)\neq1$ the numbers $T_{k,l}^j$ might significantly differ, although one expects the natural densities to be equal once we take their limits as $k,l\rightarrow \infty .$
We can compute  only some cases where $\gcd(k,l)\neq1.$
Let $p$ be an odd prime number. We are able to compute the numbers $T_{Mp,p-1}^{j}$ for $j=0,1,\dots ,p-1.$
We start with  a key lemma which is a consequence of the action of an additive group ${\mathbb Z}/p{\mathbb Z}$ on some subsets of restricted partitions.

From now on, to ease the notation, we shall skip the superscript $p$ and write $T_{k,l}^{i}$ instead of ${^p}T_{k,l}^{i}.$

Define
\begin{equation}\label{a0}
\begin{split}
S_{k,l}^{i}:= & \, \{ \,{\text{\rm the set of restricted partitions  of a number congruent to}}\quad  i\mod p \, \\
 & {\text{\rm into exactly $k$ parts  that fit into the rectangle\,\, $k\times l$ }} \}
 \end{split}
\end{equation}
Of course, we have the following equalities
\begin{equation}\label{ttt}
T_{k,l}^{i}=S_{1,l}^{i}\sqcup S_{2,l}^{i} \sqcup \dots \sqcup S_{k,l}^{i}  
\end{equation}
for any $i\in ({\mathbb Z}/p{\mathbb Z})^{*}$ and
\begin{equation}\label{ttt1}
T_{k,l}^{0}= {\bf 0}\sqcup S_{1,l}^{0}\sqcup S_{2,l}^{0} \sqcup \dots \sqcup S_{k,l}^{0}  
\end{equation}

\begin{lemma}\label{thmp}
 The following equalities hold 
\begin{equation}\label{a111}
|S_{1,p-1}^{0}|=0 \,\, {\text{\rm and}} \,\, |S_{1,p-1}^{j}|=1 \,\, {\text{\rm for}} \,\, 0\neq j\in{\mathbb Z}/p{\mathbb Z},
\end{equation}
\begin{equation}\label{a11}
|S_{k,p-1}^{i}|=|S_{k,p-1}^{j}| \,\,{\text{\rm for}}\,\,  i,j \in {\mathbb Z}/p{\mathbb Z} \,\, {\text{\rm and}} \,\,  2\leq k \leq p-1
\end{equation}

\begin{equation}\label{a122}
|S_{k,Mp}^{0}|=|S_{k,Mp}^{j}| \,\, {\text{\rm for}}\,\,  j \in {\mathbb Z}/p{\mathbb Z} \,\, {\text{\rm any}} \,\, M \geq 1 \,\,{\text{\rm and}} \,\,  1\leq k \leq p-1
\end{equation}

\end{lemma}

\begin{proof}
(\ref{a111}) follows from definition.
We will use descent with respect to $k.$ 
We assume inductively  that $|S_{m,p-1}^{i}|=|S_{m,p-1}^{j}|$  for any $ 2 \leq m < k $ and $i,j \in {\mathbb Z}/p{\mathbb Z}.$ 
To start the induction we see that a map 
\begin{equation}\label{a112}
f_{2} : S^{0}_{2,p-1}-V_{2} \rightarrow S^{2}_{2,p-1}-W_{2}     \qquad f_{2} (<{\pi }_{1},{\pi}_{2}> = <{\pi}_{1}+1, {\pi}_{2}+1>
\end{equation}
where
\begin{equation}\label{ab112}
V_{2}=\{ <{\pi}_{1},{\pi}_{2} >\in S_{2,p-1}^{0}  \quad {\pi}_{1} =p-1 \}      
\end{equation}
and
\begin{equation}\label{a113}
W_{2}= \{ <{\pi}_{1},{\pi}_{2} >\in S_{2,p-1}^{2}  \quad {\pi}_{2} =1 \}
\end{equation}
is a bijection. But $|V_{2}|=|S_{1,p-1}^{1}|$ which is easily seen by means of the assignment \linebreak $<p-1, {\pi}_{2}>\,\rightarrow \, <{\pi}_{2}>.$ Similarly, the assignment 
\nolinebreak $<{\pi}_{1},1>\,\rightarrow \,m<{\pi}_{1}>$ yields $|W_{2}|=|S_{1,p-1}^{1}|.$ We can now repeat the above argument  for $S_{2,p-1}^{2}$ and $S_{2,p-1}^{4}$ instead of $S^{0}_{2,p-1}$ and  $S^{2}_{2,p-1}$ etc.
Since $2$ (in fact any $k\neq 0$ ) is a generator of the additive group ${\mathbb Z}/p{\mathbb Z}$  we see that $|S^{0}_{2,p-1}|=|S^{j}_{2,p-1}|$ for any $j\in ({\mathbb Z}/p{\mathbb Z})^{*}.$

For any $k>2,$ we  descent from  the statement  $|S_{k,p-1}^{0}|= |S_{k,p-1}^{k}|$ to $|S_{k-1,p-1}^{1}|= |S_{k-1,p-1}^{k-1}|.$ This descent is done similarly as for the case $k=2$ by assigning    to any partition
\begin{equation}\label{a114}
 <{\pi}_{1},\dots, {\pi}_{k} > \in S_{k,p-1}^{0} -V_{k}, \quad   {\text{\rm where}} \quad V_{k}= \{ <{\pi}_{1},\dots ,{\pi}_{k} >\in S_{k,p-1}^{0}  \quad {\pi}_{1} =p-1 \}   
 \end{equation}
   the partition 
   \begin{equation}\label{a115}
   <{\pi}_{1}+1,\dots, {\pi}_{k}+1 > \in S_{k,p-1}^{k}-W_{k} \quad {\text{\rm where}} \quad W_{k}=\{ <{\pi}_{1},{\pi}_{2},\dots {\pi}_{k} >\in S_{k,p-1}^{k}  \quad {\pi}_{k} =1 \}
   \end{equation}
But  (cf. Fig.2) 
\begin{equation}\label{a5}
|V_{k}|=|S_{k-1,p-1}^{1}| \quad {\text{\rm and}} \quad |W_{k}|=|S_{k-1,p-1}^{k-1}|
\end{equation}
and by inductive hypotheses we get $|S_{k,p-1}^{0}|=|S_{k,p-1}^{k}|.$

\vspace{5mm}

\begin{tikzpicture}[inner sep=0in,outer sep=0in]

\node (n) {\begin{varwidth}{4cm}{
\begin{ytableau}\
 & & &\none& *(lightgray) \bullet & & &\none &   & &\none &*(lightgray) \bullet & & & \none &*(lightgray) \bullet &  \\
\bullet& \bullet &&\none& *(lightgray) \bullet  & *(lightgray) \bullet & *(lightgray) \bullet &\none &*(lightgray) \bullet &*(lightgray) \bullet &\none &*(lightgray) \bullet & *(lightgray) \bullet& &\none &*(lightgray) \bullet &*(lightgray) \bullet \\
*(lightgray) \bullet    & *(lightgray)  \bullet &\bullet &\none & *(lightgray) \bullet   & *(lightgray) \bullet & *(lightgray) \bullet& \none&*(lightgray) \bullet &*(lightgray) \bullet &\none & *(lightgray) \bullet& *(lightgray) \bullet& &\none &*(lightgray) \bullet & *(lightgray) \bullet \\
 *(lightgray)   \bullet  & *(lightgray) \bullet  & *(lightgray) \bullet  &\none&*(lightgray)  \bullet   &*(lightgray)  \bullet  &*(lightgray)  \bullet  &\none   & *(lightgray)  \bullet & *(lightgray) \bullet  & \none & *(lightgray)  \bullet &*(lightgray)  \bullet  & *(lightgray) \bullet  &\none &*(lightgray) \bullet & *(lightgray) \bullet\\
\end{ytableau}}\end{varwidth}};

\end{tikzpicture}
\vspace{5mm}

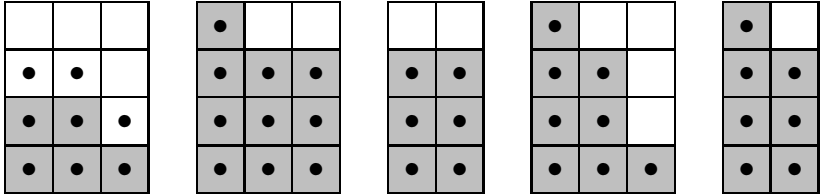
\centerline{ Fig.2 }

Again, since $k$ is a generator of ${\mathbb Z}/p{\mathbb Z}$ we see that $|S_{k,p-1}^{0}|=|S_{k,p-1}^{j}|$ for any $j\in ({\mathbb Z}/p{\mathbb Z})^{*}.$
For proving (\ref{a122}) we use essentially the same inductive argument. The difference is that this time we have $|S_{1,Mp}^{i}|=|S_{1,Mp}^{j}| $ for any $i,j \in {\mathbb Z}/p{\mathbb Z}$ and therefore we can start induction from $m=1$. 
\end{proof}

\begin{remark}\label{rck}
In Fig.2 we depicted the descent process for $k=3$ and $p=5.$ The first diagram describes the assignment  (\ref{a115}) to any partition of the form (\ref{a114}).
The second diagram depicts an element from $V_{3}$ and the third corresponding to it an element of $S_{2,4}^{1}$. Similarly, in the fourth diagram we have an element from the set $W_{3}$ and in the fifth corresponding to it  an element from $S_{2,4}^{2}.$
\end{remark}

Now, we are ready to prove
\begin{theorem}\label{therm}
For any prime number $p,$ any natural $M$  any $1\leq N \leq p-1$ and any $1\leq j \leq p-1$we have the following equalities
\begin{equation}\label{suma}
{\sum}_{k=0}^{MN}{\frak p}(Mp , N, kp+j) = {\frac{1}{p}}\left[{{Mp+N}\choose{N}}-1\right]
\end{equation}
\begin{equation}\label{48}
 {\sum}_{k=0}^{MN}\,{\frak p}(Mp , N, kp)={\frac{1}{p}}\left[{{Mp+N}\choose{N}}-1\right]+1.
 \end{equation}
 For  any prime $p$  any $1\leq N \leq p-1$ and any $0\leq j \leq p-1$ we have the following equalities
\begin{equation}\label{suma1}
{\sum}_{k=0}^{(p-1)N}{\frak p}(p-1 , N, k\equiv j \,{\text{\rm mod}} \, p) = {\frac{1}{p}}\left[{{(p-1)+N}\choose{N}}\right]
\end{equation}

\end{theorem}
\begin{proof}
The left hand side of (\ref{suma}) is equal to $T_{Mp,N}^{j}$ whereas the left hand side of (\ref{48}) equals $T_{Mp,N}^{0}.$
Notice that for $1\leq j \leq p-1$ we have
\begin{equation}\label{roz}
T_{Mp,N}^{j}=S_{1,Mp}^{j}\sqcup S_{2,Mp}^{j} \sqcup \dots \sqcup S_{N,Mp}^{j}  
\end{equation}
and for $j=0$
\begin{equation}\label{rz}
T_{Mp,N}^{0}={\bf 0} \sqcup S_{1,Mp}^{j}\sqcup S_{2,Mp}^{0} \sqcup \dots \sqcup S_{N,Mp}^{0}  
\end{equation}
 (\ref{suma}) and (\ref{48})  follow  now from (\ref{a122}) of Lemma \ref{thmp} and Corollary \ref{cor}.
 For $1\leq j \leq p-1$ we have
 \begin{equation}\label{roz1}
T_{p-1,N}^{j}=S_{1,p-1}^{j}\sqcup S_{2,p-1}^{j} \sqcup \dots \sqcup S_{N,p-1}^{j}  
\end{equation}
 and
 \begin{equation}\label{rz1}
T_{p-1,N}^{0}={\bf 0}\sqcup S_{1,p-1}^{0}\sqcup S_{2,p-1}^{0} \sqcup \dots \sqcup S_{N,p-1}^{0} . 
\end{equation}
By (\ref{a111}) and (\ref{a11}) we see now that $|T_{p-1,N}^{j}|=|T_{p-1,N}^{0}|$ and (\ref{suma1}) follows.

\end{proof}
\begin{remark}\label{rmka}
Notice that formulas (\ref{suma} )  resp. (\ref{48}) are the sums of coefficients ${\sum}_{{n\equiv r} ({\text{mod}} p)} a_{n}$ over the nonzero  resp. zero  moduli classes modulo $p$  of the Gaussian polynomial $G(Mp, N, q)={\sum}_{n=0}^{MpN}a_{n}q^{n} .$ Similarly, (\ref{suma1}) is the sum over a moduli class mod $p$ of the coefficients of the Gaussian polynomial $G(p-1, N, q)={\sum}_{n=0}^{(p-1)N}a_{n}q^{n}. $
\end{remark}

\begin{remark}\label{rrmka}
Example \ref{ex1}  shows that one can not generalise, in a straightforward way,Theorem {\ref{therm}} to the case $k=Mp, l=Np-1.$ Take $M=N=2, p=3$
and look at the values appearing in Example \ref{ex1}.
\end{remark}
 We state the following problem
\begin{problem}\label{prob}
Find relations between sums over moduli classes modulo $d$ of coefficients of Gaussian polynomials $G(k,l,q)$ for arbitrary natural numbers $k,l$ and $d.$
\end{problem}


\address{
Department of Mathematics\\
 and Physics \\
Szczecin University \\
70-415 Szczecin \\
Poland
} 
{piotrkras26@gmail.com}

\address{
Institute of Mathematics\\
Pozna\'n University of Technology\\ 
 60-965 Pozna\'n\\
Poland 
}
{jsmilew@wp.pl}

\end{document}